\documentclass[conference]{IEEEtran}
\usepackage{amsmath,amsfonts,graphicx,bm,amssymb,amsthm} 

\usepackage{amsopn}
\usepackage{subfigure}

\newcommand{\pdim}{p}
\newcommand{\ndim}{n}
\newcommand{\paino}{\sl} 
\newcommand{\hop}{\top}              
\newcommand{\covm}{\bom \Sigma} 
\newcommand{\M}{\covm} 

\newcommand{\E}{\mathbb{E}}                  
\newcommand{\bo}[1]{\mathbf{#1}}              
\newcommand{\bom}[1]{\bm{#1}}

\newcommand{\be}{\beta}

\newcommand{\iidsim}{\overset{iid}{\sim}}

\newcommand{\al}{\alpha}  
\newcommand{\Id}{\bo I}  
\newcommand{\beq}{\begin{equation}}
\newcommand{\eeq}{\end{equation}}
\newcommand{\bmat}{\begin{pmatrix}}
\newcommand{\emat}{\end{pmatrix}}

\renewcommand{\S}{{\bf S}}

\newcommand{\I}{\bo I} 

\renewcommand{\dim}{\pdim}

\newcommand{\R}{\mathbb{R}}

\newcommand{\A}{{\bf A}}
\newcommand{\B}{{\bf B}}

\newcommand{\Mn}{\M}  
\newcommand{\commat}{\bo K_\pdim}  
\newcommand{\ve}{\mathrm{vec}} 

\newcommand{\x}{{\bf x}}
\renewcommand{\u}{{\bf u}}

\newcommand{\z}{{\bf x}}

\newcommand{\Fr}{\mathrm{F}}
\newcommand{\Mor}{\bo S}
\newcommand{\Mcl}{\bo S_{\al,\be}}
\newcommand{\MSE}{\mathrm{MSE}}
\newcommand{\PDH}{\mathcal{S}}   

\newcommand{\ka}{\kappa}

\DeclareMathOperator{\tr}{tr}
\DeclareMathOperator{\Tr}{tr}

\DeclareMathOperator{\cov}{cov}

\newcounter{ctheorem}
\newtheorem{theorem}[ctheorem]{Theorem}

\newcounter{clemma}
\newtheorem{lemma}[clemma]{Lemma}

\ifCLASSINFOpdf
\else
\fi
\begin{document}
%
\title{Optimal High-Dimensional Shrinkage Covariance Estimation for Elliptical Distributions}

\author{\IEEEauthorblockN{Esa Ollila}
\IEEEauthorblockA{Department of Signal Processing and Acoustics\\
Aalto University, Finland
}}


%


\maketitle

\begin{abstract}
We derive an optimal shrinkage sample covariance matrix (SCM) estimator 
which is suitable for high dimensional problems 
and when sampling from an unspecified elliptically symmetric  distribution. 
Specifically, we derive the optimal (oracle) shrinkage parameters that obtain the minimum mean-squared error (MMSE) between the shrinkage SCM and the true covariance matrix when sampling from an elliptical distribution.  Subsequently, we show how the oracle shrinkage parameters can be 
consistently estimated under the random matrix theory regime. 
Simulations show the advantage of the proposed estimator over the conventional shrinkage SCM estimator due to Ledoit and Wolf (2004). The proposed shrinkage SCM estimator often provides significantly better performance than the Ledoit-Wolf estimator and has the advantage that consistency is guaranteed over the whole class of elliptical distributions with finite 4th order moments. 
\end{abstract}


%
\IEEEpeerreviewmaketitle

\section{Introduction}\label{sec:intro}

We consider the problem of estimating the covariance matrix based on a sample  $\x_1,\ldots, \x_n$ of independent and identically distributed (i.i.d.)  random vectors from an unspecified $\pdim$-variate distribution $\x \sim F$ with mean vector  $\E[\x] = \bo 0$  and $\pdim \times \pdim$ positive definite covariance matrix $\M= \E[\x \x^\top]$.  
The  {\paino sample covariance matrix} (SCM) $ \S=\frac{1}{\ndim}\sum_{i=1}^\ndim \z_i \z_i^\top $ 
is the most commonly used estimator of the covariance matrix,    and when random sampling from a multivariate Gaussian 
$\mathcal N_\pdim(\bo 0,\covm)$ distribution, it is also the optimal maximum likelihood estimator (MLE).  
Estimation of high-dimensional (HD) covariance matrix when the sample size $\ndim$ is smaller, or not much larger than the dimension $\pdim$, has attracted a significant research interest in recent years. Indeed since such data  problems are becoming increasingly common in finance \cite{ledoit_wolf:2004}, genomics  or classification, 
for example.  
Insufficient number of samples causes significant estimation errors in 
 the SCM. Moreover, if  $p > n$, the SCM  $\S$ is always singular, i.e., not invertible even if the true covariance matrix $\M$ is known to be positive definite and hence non-singular. 
The commonly used approach 
is then to  use shrinkage regularization as in   \cite{ledoit_wolf:2004,ollila2014regularized,zhang2016automatic,chen2010shrinkage,couillet_mckay:2014,pascal2014generalized}, for example.  

One of the most ommonly used estimator in "large $\pdim$ compared to sample size $\ndim$ problems" is the {\paino regularized SCM (RSCM)}, 
\beq \label{eq:regSCM} 
\S_{\al, \be} = \be \S + \al    \I,   
\eeq 
where 
$\al, \be >0$ denotes the {\paino shrinkage (regularization) parameters}.   Optimal RSCM estimator is often defined as one that is based on  {\paino oracle 
shrinkage parameters} that minimize the mean squared error (MSE),   
\beq \label{eq:RegSCM_oracle} 
(\al_o,\be_o)=   \underset{\al,\be >0} {\arg \min} \Big \{ \mathrm{MSE}(\S_{\al,\be}) =  \E \Big[ \big\| \M - \S_{\al,\be}  \big\|^2_{\rm F}  \Big] \Big\} ,   
\eeq 
where $\| \cdot \|_\Fr$ denotes the Frobenius matrix norm ($\| \A \|_\Fr^2 = \tr(\A^\top \A)=\tr(\A \A^\top)$ for any matrix $\A$).  
The solution  $(\al_o,\be_o)$ are called "oracle" shrinkage parameters as they will obviously  depend on the true unknown covariance matrix $\Mn$ and hence can not be used in practise. The widely popular Ledoit-Wolf (LW-)RSCM  \cite{ledoit_wolf:2004}   is based on  consistent estimators  $(\hat \al_o^{{\tiny \mbox{LW}}}, \hat \be_o^{{\tiny \mbox{LW}}})$ of $(\al_o,\be_o)$ under the random matrix theory (RMT) regime. However, more accurate finite sample estimation performance 
can be obtained by assuming that the observations are from a specific $\pdim$-variate distribution, e.g., the multivariate normal distribution,  as has been shown in \cite{chen2010shrinkage}.  In this paper, we derive  consistent estimators of the oracle shrinkage parameters $( \al_o,\be_o)$ under the RMT regime when sampling from an {\it unspecified} elliptically symmetric distribution.  Elliptical distributions (see  \cite{muirhead:1982,fang_etal:1990,ollila2012complex}) constitute a large class of distributions that include e.g., the multivariate normal distribution, generalized Gaussian and all compound Gaussian distributions as special cases. 

The  RMT regime refers to the case that 
\begin{itemize} 
\item[(R1)] 
$\ndim, \pdim \to \infty$ and $\pdim/\ndim \to c$, where $0 < c < \infty$.  
\end{itemize} 
Furthermore, we assume that the set of eigenvalues of $\Mn$ converge to a fixed spectrum, and that 
\begin{itemize} 
\item[(R2)] As $\pdim \to \infty$,   $\eta_i  = \tr(\Mn^i)/\pdim \to  \eta_i^o$,   $ 0 < \eta_i^0< \infty$ for $i=1, \ldots,4$ 
\end{itemize} 
 Our numerical examples illustrate that  the RSCM estimator that is based on  the proposed consistent estimators $(\hat \al^{Ell}_o,\hat \be^{Ell}_o)$   
 outperform its competitors, e.g., the LW-RSCM estimator, when sampling from an elliptical population. 

The paper is organized as follows.  In Section~\ref{sec:optim1} and Section~\ref{sec:optim2} we derive the optimal shrinkage parameters $(\al_o,\be_o)$ under the general assumption of sampling from any general   $\pdim$-variate distribution and  an  elliptical distribution with finite 4th order moments, respectively.  In Section~\ref{sec:estim}, consistent estimators  of $(\al_o,\be_o)$  are proposed under assumptions (R1) and (R2) when sampling from an unspecified elliptical distribution. Simulation studies of Section~\ref{sec:simul} illustrate that the proposed shrinkage estimator always outperforms the LW estimator when the samples are drawn from an elliptical population. 

{\it Notation:} Let $\PDH_p$ be the open cone of positive definite $p \times p$ symmetric matrices, and let $\I$ be the identity matrix of proper dimension,  $\ve(\cdot)$ denotes an operator that transforms a matrix into a vector by stacking the columns of the matrix, $\tr(\cdot)$ denotes the matrix trace operator,   and  $\otimes$ denote the {\paino Kronecker product}: 
for any matrix $\A$ and $\B$, $\A \otimes \B$ is a block matrix with $(i,j)$-block being equal to $a_{ij}\B$. A {\paino commutation matrix} $\commat$ is a $\dim^2 \times \dim^2$ block matrix with $(i,j)$-block 
equal to a $\dim \times \dim$ matrix that has a $1$ at entry $(j,i)$ and $0$'s elsewhere. 
It has the following important property \cite{magnus_neudecker:1999}: $\commat \ve(\A)= \ve(\A^\top) $ for any $\pdim \times \pdim$ matrix $\A$.

\section{Optimal oracle shrinkage parameters}  \label{sec:optim1}

Define {\paino scale} measures of $\M \in \PDH_\pdim$ as 
\beq \label{eq:scale} 
\eta =  \tr( \M )/\pdim \quad \mbox{and} \quad \eta_2 =  \tr( \M^2 )/\pdim . 
\eeq 
An important measure in our future developments is the following measure of sphericity \cite{srivastava2005some},  
\beq \label{eq:gamma} 
\gamma = \frac{\eta_2}{\eta^2} = \frac{\pdim \tr(\M^2)}{\tr(\M)^2}   . 
\eeq 
Statistic $\gamma$  measures how close the covariance matrix is to a scaled identity matrix.  
 It verifies $\gamma \geq 1$ and $\gamma=1$ if and only if $\Mn= c \I$ 
 for some $c>0$.

 The parameters $\eta$ and $\gamma$ are elemental in our developments. As is shown in Theorem~\ref{th:beta0_ELL}, the optimal shrinkage parameter pair $(\al_o,\be_o)$ for elliptical distributions depends on the true covariance matrix $\M$ only through $\eta$ and $\gamma$.  Simple "plug-in" estimates of $(\al_o,\be_o)$ can then be  obtained by simply replacing $(\eta, \gamma)$ with  their estimates.  Finding accurate and consistent estimators of the shrinkage parameters is then a considerably simpler task than in the general case of Theorem~\ref{th:oracle_LW}. 

Next theorem provides the expression for the oracle shrinkage parameters in the  case of sampling from an unspecified $\pdim$-variate distribution with finite $4$th order moments.

\begin{theorem}  \label{th:oracle_LW}    
Let $\{\x_i\}_{i=1}^\ndim$ denote a random sample from any $\pdim$-variate distribution (not necessarily elliptical distribution) with finite 4th order moments.  Then the oracle parameters in \eqref{eq:RegSCM_oracle} are  
\begin{align} \label{eq:beta0ver2} 
\be_o  
&=  \frac{   \pdim (\gamma-1) \eta^2  }{   \E \big[\tr \big(\S^2\big) \big]  -  \pdim \eta^2 }    
\quad \mbox{and} \quad  
\al_o =  (1 - \be_o) \eta 
\end{align} 
where $\eta$  and $\gamma$ are defined in \eqref{eq:scale} and \eqref{eq:gamma}, respectively.  
The value of MSE at the optimum is 
\beq \label{eq:MSEopt}
\MSE(\bo S_{\al_o,\be_o})  = \| \M - \eta \I \|^2_{\Fr}  ( 1- \be_o). 
\eeq 
The optimal   $\be_o$ is  always in the range $[0,1)$. 
\end{theorem} 

\begin{proof}  
It was shown in \cite[Theorem~2.1]{ledoit_wolf:2004}  that 
\begin{align} \label{eq:beta0ver2_apu} 
\be_o  &=  \frac{ \| \M -  \eta  \I \|^2_{\rm F}}{  \| \M -  \eta  \I \|^2_{\rm F}  +  \E\big[ \| \S -  \M \|^2_{\rm F}\big] } 
\end{align} 
and $\al_o =  (1 - \be_o) \eta$. The form of $\be_o$ in \eqref{eq:beta0ver2_apu} implies that $\be_o \in [0,1)$.  
We now show that \eqref{eq:beta0ver2_apu} can be expressed in the form \eqref{eq:beta0ver2}. 
First, we observe that 
\begin{align} 
 a_1 &= \E \big[ \| \Mor - \Mn \|^2_{\Fr} \big] =   \E \big[\tr \! \big(\Mor^2\big) \big] - 2 \E \big[ \tr\big(\Mor \Mn\big)\big] + \tr\big(\Mn^2\big) \notag 
 \\ 
 &=  \E \big[\tr \big(\Mor^2\big) \big]  -  \tr(\Mn^2) \label{eq:varrho}
\end{align}
where we used that $\E[\tr(\Mor \Mn)] = \Tr(\E[\Mor]\Mn) =  \tr(\Mn^2)$. 
The numerator of $\be_o$ in \eqref{eq:beta0ver2_apu} is 
\begin{align} 
a_2 &=\| \Mn -  \eta  \I \|^2_{\rm F} = \tr(\Mn^2) - (1/\pdim) \big\{\tr(\Mn) \big\}^2   \notag \\ 
&= \pdim( \eta_2 - \eta^2)  =  \pdim(\gamma-1) \eta^2  \label{eq:th_oracle_LW_apu}
\end{align} 
which shows that denominator of $\be_o$ is $a_1  + a_2 = \E \big[\tr\! \big(\Mor^2\big) \big]  - (1/\pdim) \big\{\tr(\Mn) \big\}^2  =  \E \big[\tr\! \big(\Mor^2\big) \big]  - \pdim \eta^2$.   These expressions for numerator and denominator of $\beta_o$ yield the assertion \eqref{eq:beta0ver2}  for $\be_o$. 
Write $L(\al,\be)=\E \big[ \| \Mcl  - \Mn \|_{\rm F}^2 \big]$ for the MSE. Note that 
\begin{align} 
L(\al,\be) 
&=  \E \Big[ \big \| \al \I  + \be( \Mor - \Mn) - (1-\be) \Mn  \big\|_{\rm F}^2 \Big]  \notag
 \\ &= \al^2 \pdim + \be^2  a_1  + (1-\be)^2 \eta_2 \pdim   - 2\al(1-\be) \eta \pdim.  \label{eq:oracle_LW_apu} 
\end{align} 
The MSE at the optimum is 
\begin{align*} 
\MSE(\bo S_{\al_o,\be_o})   &= L((1-\be_o)\eta,\be_o)\\ &= \be_o^2 a_1 +  (1-\be_o)^2 \eta_2 \pdim  - (1-\be_o)^2 \eta^2 \pdim  \\ 
&= (1-\be_o^2)^2 \underbrace{ \{ \pdim(\eta_2 - \eta^2) \}} _{ \mbox{$=a_2$ by \eqref{eq:th_oracle_LW_apu}}}  + \be_o^2 a_1 \\ &= (1-\be_o)^2a_2 + (1-\be_o) \be_o a_2  \\
&= (1-\be_o) a_2,
\end{align*} 
where the 3rd identity follows as $\be_o a_1 = (1-\be_o) a_2$. This completes the proof.  
\end{proof}

Theorem has important implications.  First, since $\al_o = (1-\be_o)\eta$  is determined by the value of $\be_o \in [0,1)$, the optimal RSCM can be expressed simply 
as
\[ 
\bo S_{\al_o,\be_o}  =   \be_o \frac{1}{n} \sum_{i=1}^\ndim  \z_i \z_i^\top      +  (1-\be_o)\eta   \I.   
\] 
Since $\hat \eta=\tr(\S)/\pdim$ is a consistent estimator of $\eta=\tr(\M)/\pdim$ both in the conventional (fixed $\pdim$) and RMT asymptotic regime, we need  to simply  focus on finding a consistent estimator $\hat \be_o$ of $\be_o$. Consistent estimator of $\al_o$ is determined simply as 
$\hat \al_o = (1-\hat \be_o) \tr(\S)/\pdim$.   

Ledoit and Wolf \cite{ledoit_wolf:2004} showed  that 
the following estimate 
\[
\hat \be_{{\tiny \mbox{LW}}}^* =  1 -  \frac{ \sum_{i=1}^\ndim \| \x_i \x_i^\top - \S \|_\Fr^2 }{\pdim \ndim^2 (\hat \gamma-1)}    
=  1 -  \dfrac{   \frac{1}{\ndim\pdim} \sum_{i=1}^\ndim \| \x_i \|_2^4  -  \hat \eta_2}{\ndim (\hat \gamma-1)}    
\]
where $ \hat \gamma=  \hat \eta_2/\hat \eta=p\tr(\S^2)/\tr(\S)$ and $\hat \eta_2 =  \tr( \S^2 )/\pdim$, converges to $\be_o$ in \eqref{eq:beta0ver2} in   probability   under RMT regime (R1) and  (R2) 
when sampling from a distribution $\x \sim F$ with finite $4$th-order moments. 
The authors of \cite{ledoit_wolf:2004}  then proposed to estimate the shrinkage parameters using  
\begin{align*} 
\hat \be_0^{{\tiny \mbox{LW}}} &= \max(0, \hat \be_{{\tiny \mbox{LW}}}^*)  \quad \mbox{and} \quad \hat \al_0^{{\tiny \mbox{LW}}}   = (1-  \hat \be_0^{{\tiny \mbox{LW}}}) \tr(\S)/\pdim , 
\end{align*} 
where the $\max$ constraint ensures that the final estimate  remains on the interval $[0,1]$. The RSCM based on the above penalty parameters 
is referred hereafter as {\paino LW-RSCM estimator}.   

\section{Optimal oracle shrinkage parameters: the elliptical case} \label{sec:optim2}

Assume now that  $\x_1,\ldots, \x_n$ are independent and identically distributed (i.i.d.)  random vectors from 
a centered elliptical distribution with  mean vector $\E[\x]=\bo 0$ and positive definite covariance matrix $\M = \E[ \x \x^\top] $, denoted $ \mathcal E_\pdim(\bo 0,\M,g)$.  For a review of elliptical distributions, see  \cite{muirhead:1982,fang_etal:1990,ollila2012complex}.  
The probability density function (p.d.f.) of   $\x \sim \mathcal E_\pdim(\bo 0,\M,g)$ is 
\[
f(\z) = C_{\pdim,g} 
  |\M|^{-1/2} g\big( \z^\top \M^{-1} \z\big)
 \] 
where  $g: [0,\infty) \to [0,\infty)$ is  a fixed function, called the  {\paino density generator}, that is independent of $\x$ and $\M$, and $C_{\pdim,g}$ is a normalizing constant ensuring that $f(\z)$ integrates to 1. Let $g$ be defined so that $\M$ represents the covariance matrix  of $\x$.  For example, the $\pdim$-variate Gaussian distribution, denoted $\x \sim \mathcal N_\pdim(\bo 0,\M)$, is a member in this class with density generator $g(t)=\exp(-t/2)$.  
As earlier in Theorem~\ref{th:oracle_LW},  we   
assume that the elliptical population possesses finite 4th-order moments. 

Recall that the {\paino kurtosis} of a zero mean random variable $x$  is defined as 
\[
 \mathrm{kurt}(x) = \frac{\E[ x^4]}{ (\E[x^2])^2} - 3 .
\]
The {\paino elliptical kurtosis parameter} \cite{muirhead:1982}  $\ka$ of a random vector $\x=(x_1,\ldots,x_\pdim)^\top \sim \mathcal E_\pdim(\bo 0,\M,g)$ is defined as 
 \beq  \label{eq:kappa} 
 \kappa =  \frac{ \E[r^4] }{\pdim(\pdim+2)}  -1  =  \frac 1  3 \cdot  \mathrm{kurt}(x_i),  
 \eeq 
 where  $r$ denotes the (2nd order) modular variate of the elliptical distribution,  defined as $r= \sqrt{ \x^\top \M^{-1} \x}$.  
 The elliptical kurtosis 
 shares properties similar to kurtosis of a real random variable.  Especially, if    
$\x \sim \mathcal N_\pdim(\bo 0,\M)$, then $\ka=0$.  This is  obvious since the marginal distributions are Gaussian and  hence $\kappa = (1/3) \, \mathrm{kurt}(x_i)=0$.  
Another way to derive this is by noting that $ r^2=\x^\top \M^{-1} \x  \sim \chi^2_\pdim$ 
 and hence 
 $\E[r^4] = \pdim(\pdim+2)$. 
    The importance of elliptical kurtosis parameter $\kappa$ is due to the fact that the $\pdim^2 \times \pdim^2$ covariance matrix of $\ve(\S)$ can be expressed as 
      \cite{muirhead:1982}:
 \begin{gather}
 \cov( \ve(\S))  =  \notag \\  
  \frac{ (1+\ka)}{\ndim}  (\I + \commat)  (\M \otimes \M) +  \frac{\ka}{\ndim} \,  \ve(\M) \ve(\M)^\top, \label{eq:cov_vecS} 
 \end{gather} 
 where $\commat$   denotes the commutation matrix defined in the Introduction. Thus the elliptical kurtosis parameter $\ka$ along with the true covariance matrix $\M$ provide a complete description of the covariances between elements $S_{ij}$ and $S_{kl}$   of the SCM $\S$.  

In the next Lemma we derive the MSE of the SCM. 

\begin{lemma} \label{lem:NMSE}   Let $ \{ \x_i \}_{i=1}^\ndim \iidsim \mathcal E_\pdim(\bo 0,\M,g)$, where $\M = \cov(\x_i)$
and 4th-order moments exists. 
Then the MSE of $\S$ is 
\[
\MSE(\S)=  \frac \pdim \ndim \cdot  \eta^2  \Big\{  \kappa( 2 \gamma +  \pdim) +   \gamma+\pdim \Big\},
\]
and the {\paino normalized mean squared error (NMSE)} is 
\[
\mathrm{NMSE}(\S)  =   \frac{\E \big[ \| \S  - \M \|_\Fr^2 \big] }{ \| \M \|^2_\Fr} =  \frac 1 \gamma \cdot  \frac 1 \ndim  \Big\{  \kappa( 2 \gamma +  \pdim) +  \gamma+\pdim \Big\}.
\]
Furthermore, 
\begin{align*} 
\E [ \tr(\S^2)]  &= \MSE(\S) + \pdim \eta_2. 
 \end{align*} 
Above  $\eta, \gamma$ and  $\ka$ are defined in \eqref{eq:scale}, \eqref{eq:gamma} and \eqref{eq:kappa}, respectively. 
\end{lemma} 

\begin{proof} 
Since $\S$ is unbiased, so $\E[\S] = \M$, it holds that 
\beq \label{eq:lem:NMSE:apu1} 
\mathrm{MSE}(\S) =  \tr \{ \cov( \ve(\S)) \} ,
\eeq 
where $\cov(\ve(\S))$ has the expression stated in \eqref{eq:cov_vecS}.  
Then recall the following results: 
$\tr(\A \otimes \B) = \tr(\A) \tr(\B) $,  $\tr \{  \ve(\A) \ve(\B)^\top \} = \tr(\A \B)$ for any square matrices $\A$ and $\B$ of same order; see e.g., \cite{magnus_neudecker:1999}. 
These imply that 
\beq \label{eq:lem:NMSE:apu2} 
\tr( \M \otimes \M)= \tr(\M)^2, \, \,  \tr \{\ve(\M) \ve(\M)^\top \} = \tr(\M^2).
\eeq 
It is also easy to  show that  
\beq \label{eq:lem:NMSE:apu3} 
\tr\big\{ \commat (\M \otimes \M) \big\} = \tr(\M^2)
\eeq
by recalling the definition of the commutation matrix and  the property $\tr(\A \otimes \B) = \tr(\A) \tr(\B)$. 
 Using  \eqref{eq:lem:NMSE:apu1} -   \eqref{eq:lem:NMSE:apu3}, then yield the stated expression for $\MSE(\S)$. The expression for NMSE is obtained by dividing $\MSE(\S)$  by $\tr(\M^2) = \pdim \eta_2$. The last results follows as
\begin{align*} 
\MSE(\S) &=\E[ \| \S - \M \|_\Fr^2 ] = \E[ \tr\{ (\S - \M) (\S - \M)\}] \\ &= \E[ \tr(\S^2) - 2 \tr(\S \M) + \tr(\M^2)]  \\ &= \E[\tr(\S^2)] - \pdim \eta_2
\end{align*} 
by using that $\E[\tr(\S \M)] = \tr(\E[\S] \M) = \tr(\M^2) = \pdim \eta_2$.
\end{proof}  

Next theorem states that the oracle parameters derived in Theorem~\ref{th:oracle_LW} can be written in a much simpler form when sampling from an 
elliptically symmetric distribution.  
 
 \begin{theorem}  \label{th:beta0_ELL} Let $\{\x_i\}_{i=1}^\ndim \iidsim  \mathcal E_\pdim(\bo 0, \M, g)$ 
and assume that elliptical population possesses finite    4th-order moments.   
Then the oracle  parameters $(\al_o,\be_o)$  that minimize the MSE 
are  
\begin{align*} 
\be_o^{Ell} &= \dfrac{ \gamma-1}{   (\gamma-1) \, + \, \gamma \cdot  \mathrm{NMSE}(\S) }  \\
&= \dfrac{ \gamma-1 }{   \gamma-1 +  (1/\ndim)\{  \kappa( 2 \gamma +  \pdim) +  \gamma+\pdim \} }  
\end{align*}
and $\al_o^{Ell} = (1-\be_o^{Ell}) \eta$.   
\end{theorem} 

\begin{proof} Using Lemma~\ref{lem:NMSE}, the denominator of $\be_o$ is 
\begin{align*} 
\E\big[&\tr\! \big(\S^2\big) \big]  -  \pdim \eta^2  \\ 
&= \MSE(\S) +\pdim  \eta_2 - \pdim \eta^2 \\ &= \pdim \eta^2 \{ \MSE(\S)/(\pdim \eta^2) + \gamma-1    \}  \\ 
&=  \pdim \eta^2 \{  \gamma \cdot \mathrm{NMSE}(\S)  +  \gamma-1 \}, 
\end{align*} 
where the last idenitity follows as $\mathrm{NMSE}(\S)= \MSE(\S)/\| \M\|_F^2 = \MSE(\S)/(\pdim\eta_2)$  and recalling that $\gamma=\eta_2/\eta^2$. 
Substituting this expression into \eqref{eq:beta0ver2} yields the first assertion for  $\be_o$. The second assertion follows by 
recalling the expression for $\mathrm{NMSE}(\S)$   from Lemma~\ref{lem:NMSE}.  
\end{proof} 

It is not surprising that $\beta_o$ and hence also $\al_o$ depend on the functional form of the elliptical distribution (i.e., on density generator $g$) only via elliptical kurtosis parameter $\ka$.  Specifying the elliptical distribution (e.g., Gaussian, $t$-distribution, etc), also specifies the value of  $\kappa$.   
For example, when sampling from the Gaussian distribution, the elliptical kurtosis parameter is $\kappa=0$, 
but since we do not assume any particular elliptical distribution, we need to 
find a  consistent estimator of the elliptical  kurtosis parameter  $\hat \kappa$ as well.  

\section{Consistent estimation of the oracle parameters} \label{sec:estim} 

Let $\{\x_i\}_{i=1}^\ndim \iidsim \mathcal E_\pdim(\bo 0, \M , g)$, where $\cov(\x)=\M$ and assume that the $4$th-order moments exists. 
In this section, we address the important topic of how to obtain consistent estimators of the unknown parameters $\eta, \gamma$ and $\kappa$. 

First we recall that the sample sign covariance matrix,   
defined 
as 
\[
\S_{sgn}  = \frac{1}{\ndim} \sum_{i=1}^\ndim \frac{\x_i \x_i^\top}{\| \x_i \|^2} , 
\]
is well-known to be highly robust although it is not a consistent estimator of the covariance matrix \cite{croux2002sign}.  However,  the following result from \cite[Lemma~4.1]{zhang2016automatic} shows that it can be used to estimate the parameter $\gamma$. 

\begin{lemma}  \label{lem:eta2_ell1}  Let $\{\x_i\}_{i=1}^\ndim \iidsim \mathcal  E_\pdim(\bo 0,\M,g)$. Then  
\beq \label{eq:gammahat}
\hat \gamma = \pdim \tr \big(\S_{sgn}^2 \big) - (\pdim/\ndim) 
\eeq 
is a consistent estimator of $\gamma = \pdim \tr(\M^2)/\tr(\M)^2$ under assumption (R1) and (R2). 
\end{lemma}

Note that $\hat \gamma$ is a robust and distribution-free estimator of $\gamma$. 
The optimum parameter $\beta_o^{Ell}= \beta_o^{Ell}(\gamma, \ka)$ depends on $\gamma$ and $\ka$. Hence a {\paino plug-in estimator}, 
\[ 
\hat \beta_o^{Ell}= \beta_o^{Ell}(\hat \gamma, \hat \ka),
\] 
where $\hat \gamma$ and $\hat \ka$ are consistent estimators of $\gamma$ and $\kappa$, is a consistent estimator of  $\beta_o$ as well. 
  A natural estimate of $\ka$  is the  conventional sample average, 
\beq \label{eq:kappahat}
\hat \kappa =  \max \Big( - \frac{2}{p+2} \, ,  \,  \frac{1}{3 \pdim} \sum_{j=1}^\pdim  \hat k_j \Big), 
\eeq 
where   $\hat k_j =  m^{(4)}_j/ \big(m_j^{(2)} \big)^2 - 3$   is the sample kurtosis of the $j$th variable  and 
$m^{(q)}_j = \frac 1 \ndim \sum_{i=1}^\ndim  (x_{ij})^q$ 
denotes the $q$th order sample moment, 
 $j=1,\ldots, \pdim$.   Above  the $\max$ constraint ensures that the final estimate  $\hat \ka$ does not exceed the theoretical lower bound \cite{bentler1986greatest},  $-2/(p+2)$ of elliptical kurtosis parameter $\ka$.  
 The estimate $\hat \kappa$ is a consistent estimator of the elliptical kurtosis $\kappa$ both in the conventional and RMT regime.

 
 
We can now define the {\paino Ell-RSCM estimator} as the regularized SCM  based on the following estimated optimal shrinkage parameters 
\begin{align} 
 \hat \beta_o^{Ell} &=  \max \left(0, \dfrac{T}{  T +  (1/\ndim)\{ \hat \kappa (  2 \hat \gamma +  \pdim) +   \hat \gamma + \pdim \} } \right)   \label{eq:Ell-RSCM}  \\ 
 \hat \al_o^{Ell} &= (1-  \hat \be_0^{Ell}) \tr(\S)/\pdim  \notag
\end{align} 
where $T = \hat \gamma -1$ and $\hat \gamma$  and $\hat \kappa$ are defined in \eqref{eq:gammahat} and  \eqref{eq:kappahat}, respectively,


\section{Simulation study} \label{sec:simul} 

We conduct a small simulation study to investigate the performance of RSCM estimators in terms of their finite sample NMSE. 
Each simulation is repeated 10000 times and the NMSE is computed (averaged of Monte-Carlo runs) for each  RSCM estimator.  Theoretical oracle MSE value derived in  \eqref{eq:MSEopt} and normalized by $\| \M \|^2_{\Fr}$ is used as a benchmark lower bound for empirical NMSE values.  This is shown in the figures as solid black line.


\subsection{AR(1) covariance matrix} 


 In the first experiment, an autoregressive covariance structured  is used. We let $\M$ be the covariance
matrix of a Gaussian AR$(1)$ process, 
\[
[\M]_{ij}  = \varrho^{|i-j|},  \quad r \in (0,1) . 
\]
Note that $\M$ 
verifies $\eta = \tr(\M)/\pdim = 1$. When $\varrho$ is close to $0$, then $\M$ is 
close to an identity matrix and when $\varrho$  tends to $1$, $\M$ tends to  a singular matrix of rank 1. 
Thus the theoretical value $\be_o$ is close to $0$ for small values of $\varrho$, i.e., when the true covariance matrix is close to the target $\I$, 
and $\be_0 \approx 1$ for $\varrho$  close to $1$. Dimension is  fixed at $\pdim = 100$  and  $\ndim$ is allowed to vary from $0.2 \cdot \pdim $ to $1.2 \cdot \pdim$. 

Figure~\ref{fig:AR1_p100}  depicts the NMSE performance  
when the samples are drawn from a Gaussian distribution (upper panel) and 
 a multivariate $t_\nu$-distribution with $\nu=8$ degrees of freedom (lower panel).  
Several conclusions can be drawn from these figures. First, when $\varrho=0.1$ and thus $\M$ is close to the shrinkage target matrix $\I$,  Ell-RSCM 
estimators outperform the LW-RSCM estimator.  Especially, when the ratio $\ndim/\pdim$ is small (i.e., $\pdim$ larger than $\ndim$), we observe the largest performance differences in favor of Ell-RSCM. 
Second, when the true $\M$ starts to deviate significantly from the identity target matrix $\I$ (i.e., $\varrho=0.4$), LW-RSCM and Ell-RSCM estimator have similar performance especially for large values of $\ndim/\pdim$.  
 Third, when the samples are drawn from $t_8$-distribution, the performance of LW-RSCM estimator is seen to deterioritate in comparison to the proposed Ell-RSCM estimator.  Indeed very large differences are witnessed in NMSE between the estimators especially when $\ndim/\pdim < 0.5$.  


\begin{figure}
\subfigure[$\varrho=0.1$]{\includegraphics[width=0.24\textwidth]{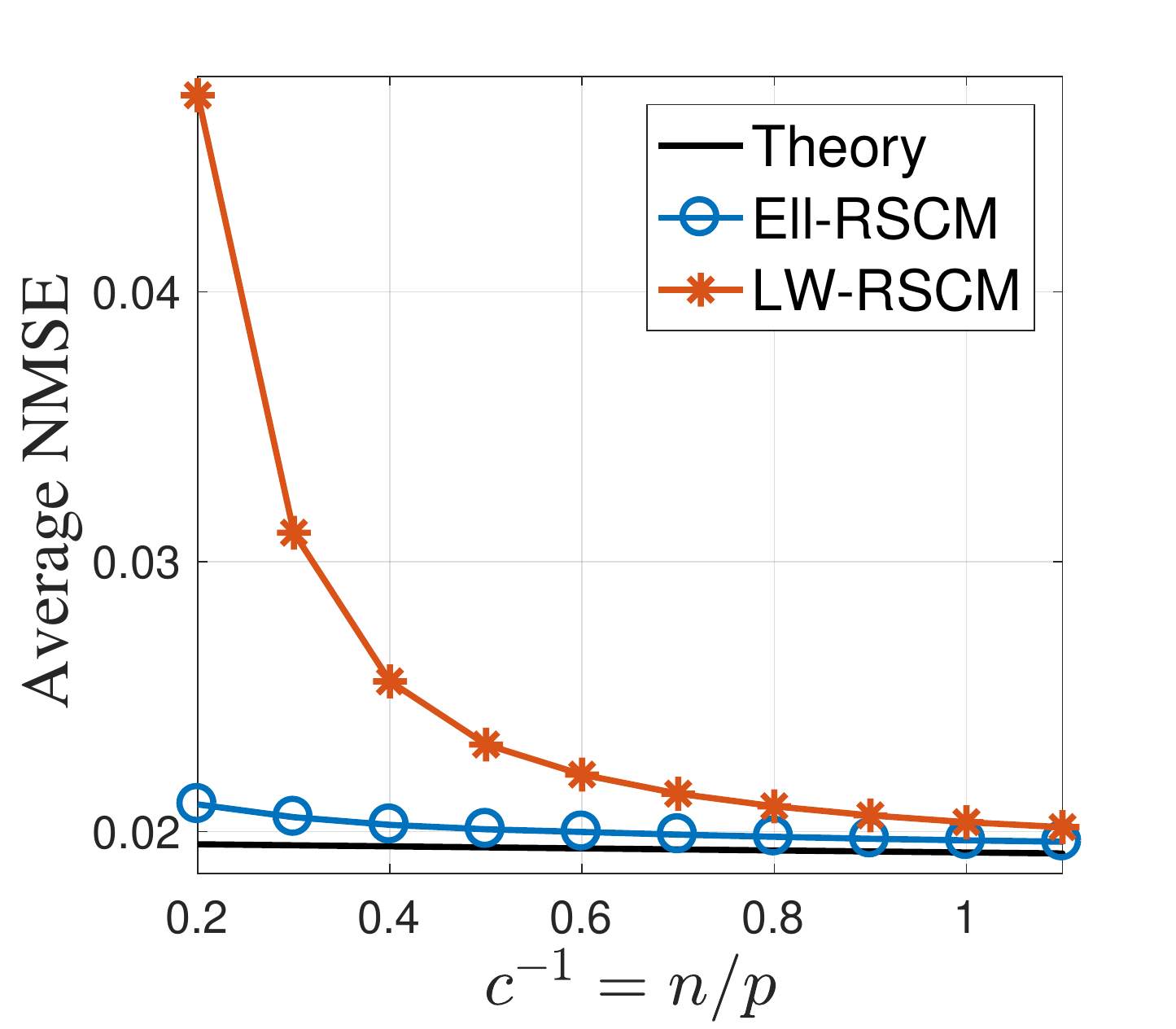}}
\subfigure[$\varrho=0.4$]{\includegraphics[width=0.24\textwidth]{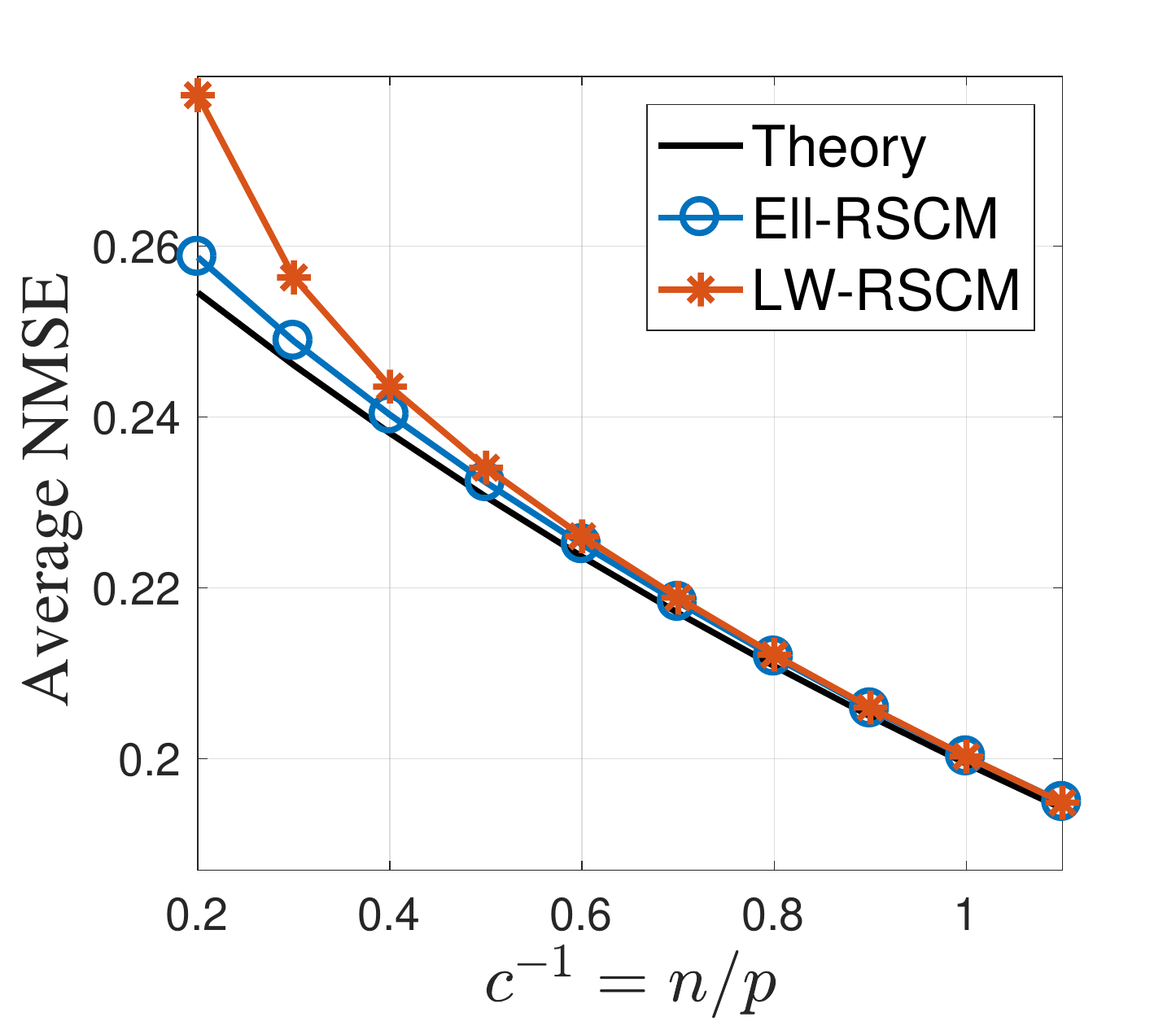}}
\subfigure[$\varrho=0.1$]{\includegraphics[width=0.24\textwidth]{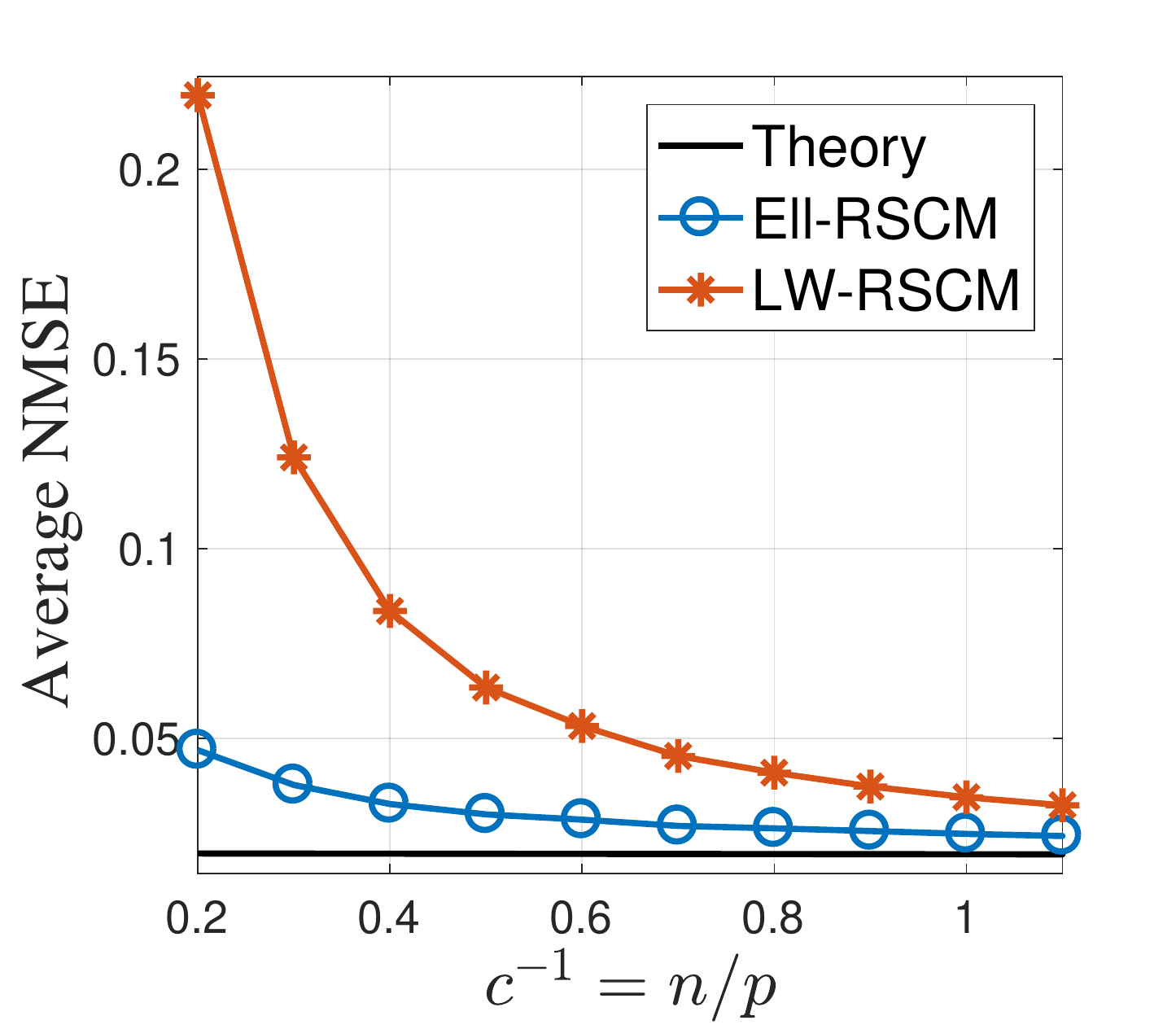}}
\subfigure[$\varrho=0.4$]{\includegraphics[width=0.24\textwidth]{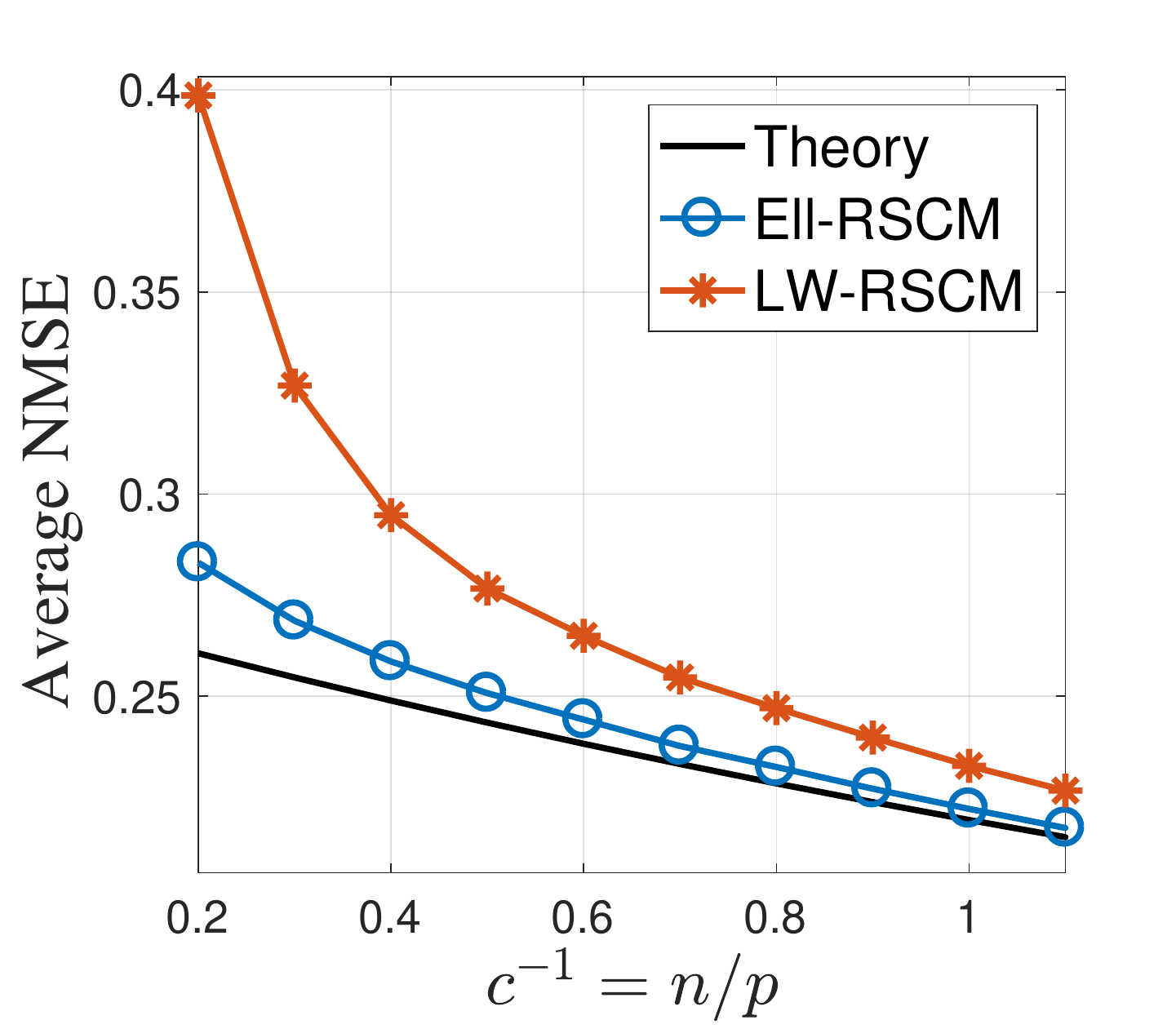}}
\caption{AR(1) process: Comparison of covariance estimators when $\pdim = 100$ and  $\varrho \in \{ 0.1, 0.4\}$ and the samples are from Gaussian distribution (upper panel) 
and $t_\nu$-distribution with $\nu=8$ degrees of freedom (lower panel).}   \label{fig:AR1_p100}
\end{figure}

\subsection{Largely varying spectrum} 

Our next study follows the set-up in \cite{zhang2016automatic} in which $\M$ has one
(or a few) large eigenvalues.
In the first set-up,  $\M$ is a diagonal matrix of size 
$50 \times 50$, where $m$ eigenvalues are equal to $1$ and the remaining $50-m$ eigenvalues are 0.01.
For the case $\ndim = \pdim=50$, Figure~\ref{fig:tengsim1} depicts the NMSE as a function of $m$  
when sampling from a $t_\nu$ distribution with $\nu=8$ degrees of freedom. 
Ell-RSCM has excellent performance as its NMSE curve is essentially overlapping with  the theoretical NMSE curve. 
LW-RSCM estimator is performing poorly for all values of $m$ except at the extremes, i.e, when $m$ is either small or  large, in which case the covariance matrix $\M$  is close to an (scaled) identity matrix. 
  

\begin{figure}
\centerline{\includegraphics[width=0.45\textwidth]{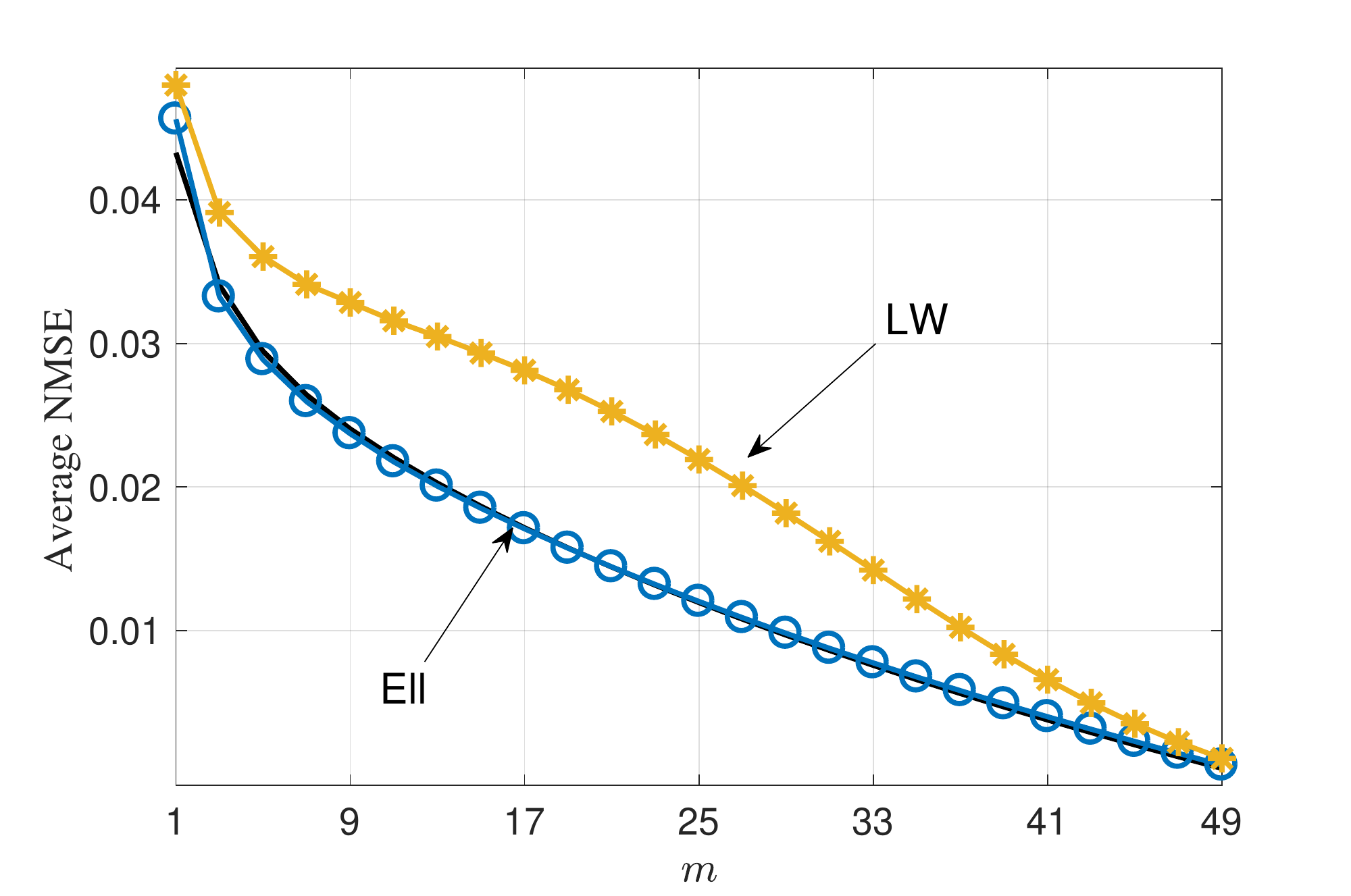}}
\vspace{-0.2cm}
\caption{The covariance matrix $\M$ has $m$ eigenvalues equal to $1$ and $50-m$ eigenvalues equal to $0.01$. The samples are from $t_\nu$-distribution with $\nu=8$ degrees of freedom and $\ndim = \pdim = 50$. }   \label{fig:tengsim1}
\end{figure}

Next simulation set-up considers a very challenging scenario in which the spectrum of $\M$ 
consists of several different eigenvalues. 
We consider the case that $\pdim=100$ and the covariance matrix  $\M$ has 30 eigenvalues equal to $100$, $40$ eigenvalues equal to $1$
and 30 eigenvalues of $0.01$. Samples are drawn from $t_\nu$ distribution with $\nu=8$ degrees of freedom. 
The NMSE curves shown in Figure~\ref{fig:tengsim1} illustrate the huge advantage of the proposed Ell-RSCM  over the LW-RSCM estimator. In fact, in this scenario the LW estimator fails and it assigns  $\hat \beta_o= 0$ for all values of $\ndim$. 
 Again the Ell-RSCM estimator  reaches  near oracle performance and thus there is not much space for improvements.
  
\begin{figure}
\centerline{\includegraphics[width=0.45\textwidth]{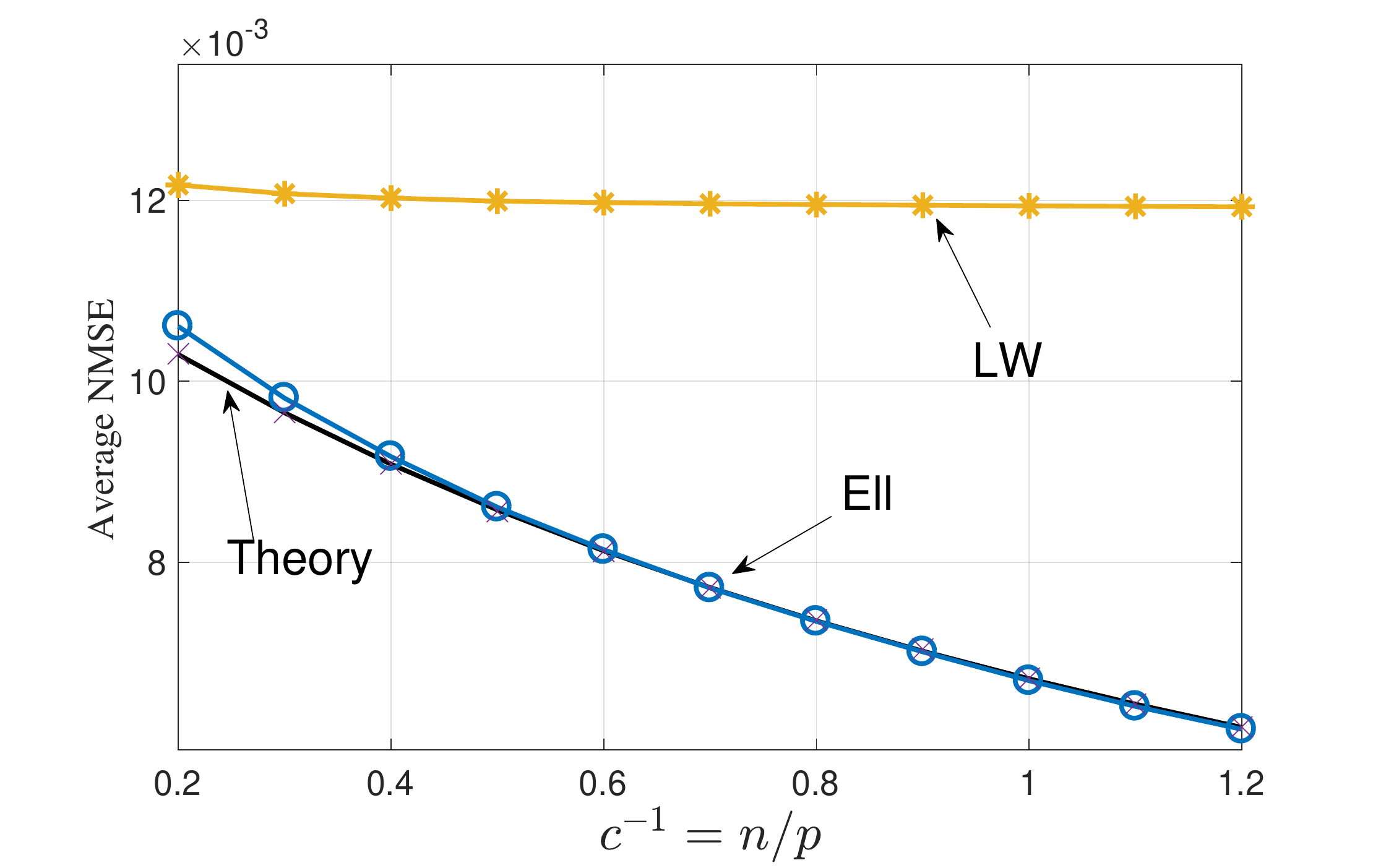}}
\vspace{-0.2cm}
\caption{ The covariance matrix  $\M$ has 30 eigenvalues equal to $100$, $40$ eigenvalues equal to $1$
and 30 eigenvalues equal to $0.01$.  The samples are from $t_\nu$-distribution with $\nu=8$ degrees of freedom and  $\pdim=100$.}   \label{fig:tengsim2}
\end{figure}

\section{Conclusion}

We proposed an optimal regularized sample covariance matrix estimator, called Ell-RSCM estimator,   
which is suitable for high-dimensional problems and when sampling from an unspecified elliptically symmetric  distribution. The estimator is based on consistent estimators (under RMT regime) of the optimal shrinkage parameters that minimize the MSE. It smartly exploits  elliptical theory such  as the knowledge of the form of  MSE of the SCM when sampling from an elliptical population.  Our simulation studies illustrated the advantage of the proposed Ell-RSCM over the Ledoit-Wolf (LW-)RSCM estimator.  The performance differences were often significant.






%



\end{document}